\documentclass[letterpaper, 10 pt, conference]{ieeeconf}  
\pdfoutput=1

\usepackage[utf8]{inputenc}
\usepackage{amsmath,amssymb,amsfonts}
\usepackage{algorithmic}
\usepackage{graphicx,epstopdf}
\usepackage{cite}
\usepackage{newtxmath}
\usepackage{color}
\usepackage{algorithm}
\makeatletter
  \renewcommand{\ALG@name}{Implementation procedure}
  \makeatother

\newtheorem{thm}{Theorem}
\newtheorem{lemma}{Lemma}
\newtheorem{rem}{Remark}

\newtheorem{defin}{Definition}

\IEEEoverridecommandlockouts                              

\title{\LARGE \bf
Adaptive flexibility function in smart energy systems: A linearized price-demand mapping approach
}

\author{Seyed Shahabaldin Tohidi, Henrik Madsen, Georgios Tsaousoglou and Tobias K. S. Ritschel
\thanks{$^{1}$Shahab Tohidi, Henrik Madsen, Georgios Tsaousoglou, and Tobias K. S. Ritschel are with Department of Applied  Mathematics and Computer Science, Technical University of Denmark, DK-2800 Kgs. Lyngby, Denmark
        {\tt\small \{sshto, hmad, geots, tobk\}@dtu.dk}}%
}

\begin{document}

\maketitle
\thispagestyle{empty}
\pagestyle{empty}

\begin{abstract}

This paper proposes an adaptive mechanism for price signal generation using a piecewise linear approximation of a flexibility function with unknown parameters. In this adaptive approach, the price signal is parameterized and the parameters are changed adaptively such that the output of the flexibility function follows the reference demand signal provided by the involved aggregator. This is guaranteed using the Lyapunov stability theorem. The proposed method does not require an estimation algorithm for unknown parameters, that eliminates the need for persistency of excitation of signals, and consequently, simplifies offering the flexibility services. Furthermore, boundedness of the price signal is ensured using a projection algorithm in the adaptive system. We present simulation results that demonstrate the price generation results using the proposed approaches.
\end{abstract}

\section{INTRODUCTION}

Recent reports about climate change state a significant increment of earth surface temperature, known as global warming. State of the climate in Europe \cite{WMO22} establishes that Europe is the fastest-warming of the six regions defined by World Meteorological Organization (WMO). Since the 1980s, Earth's temperature has increased at a rate of +0.5${^{\circ}}$C per decade, which is more than twice the global average, in Europe. As reported in \cite{Rashit22}, mean temperature of arctic region ($>$60${^{\circ}}$ latitude) was 0.71${^{\circ}}$C higher than the average of the preceding decade. Therefore, it is necessary to look for an energy management solution to alleviate the energy consumption and shift to the green resources of energy.

Expanding the renewable energy sources, like solar and wind power, decentralizes the energy production. Consequently, the demand must be adjusted to meet the existing generated power \cite{Jaume23}. In short, the energy system is in a transition from a centralized system with relatively few power generation facilities to a decentralized system where the balance is ensured by demand-side response and local intelligent systems \cite{Rongling23, Rongling22, George23}. 

Demand Side Management (DSM) consists of various control strategies for load shifting, peak shaving, or demand reduction \cite{Pean19}. This requires the demand side profile to be flexible, that is, it should be capable of managing its demand and generation based on user needs, grid balancing and local climate conditions \cite{Jensen17, Lund15}. For example, flexibility potential of thermal dynamics of a building is dependent on its inherent thermal mass and storage options such as water tanks, along with the Heating, Ventilation and Air Conditioning (HVAC) system. Advanced control design has shown to have great potential for activating this flexibility potential \cite{Pean19, Rongling20}.


The flexibility function, a mapping between price and demand over time in a price-responsive system, is proposed as a minimum interoperability mechanism (MIM) between the aggregator and the individual flexible assets (e.g., buildings) \cite{Rune18}. A generalized version of the flexibility function involving a nonlinear mapping between price and demand is provided in \cite{Rune20}. Specifically, the mapping describes the temporal evolution of the energy demand in response to changes in the energy price \cite{Domink20}. Therefore, it can be used for demand scheduling and load shifting.

The dynamics of a given energy system changes over time due to gradual deterioration, seasonal changes (e.g., in ambient temperature), consumer behavior (e.g., during holidays), etc. Therefore, the dynamics of the price-demand relationship changes as well. This motivates the design of a mechanism that accounts for such dynamic variations. The issue of parametric uncertainty can be mitigated using either passive or active approaches. Passive methods are based on robust fixed-structure control systems considering bounded parametric uncertainty \cite{bemporad07,zeilinger10, Doyle1998, liu22}. In contrast, active methods are based on adaptive control methods that adjust the control law based on the changes in system parameters \cite{GHOLAMI21, Lemos14, Pierre87, Tohidi20, Annaswamy12}.



This paper proposes an adaptive flexibility controller capable of updating the control law, i.e. the price signal, based on the changes in the price-demand dynamics. To the best of our knowledge, existing methods based on the flexibility function do not account for parametric uncertainty. In this adaptive approach, the price signal is parameterized and the parameters are changed adaptively such that the demand closely follows the reference demand signal provided by an aggregator. Another benefit of employing this approach is that it is not based on system identification methods. Hence, it does not require any persistency of excitation assumption on the input signals \cite{aastrom13}. Moreover, a projection algorithm has been employed to confine the adaptive parameters within a prespecified compact set to guarantee the boundedness of the price signals \cite{Tohidi22a, Lavretsky13}. The adaptation capability simplifies offering the flexibility services e.g. in a plug-and-play manner and without the need to conduct a manual, customized modeling-and-control study for each resource separately.

This paper is organized as follows. Section \ref{sec:FF} provides an overview of the flexibility function considered in this work. Section \ref{sec:linFF} presents a linearized version of the flexibility function. Section \ref{sec:OPG} provides an optimal control signal assuming that all parameters are known. Considering unknown parameters, Section \ref{sec:AFF} proposes an adaptive flexibility function mechanism while ensuring the boundedness of the control signal. Section \ref{sec:sim} presents the simulation results, and a summary is provided in Section~\ref{sec:conc}.


\section{FLEXIBILITY FUNCTION}\label{sec:FF}

Nonlinear dynamics of price-demand relationship proposed in \cite{Rune20} is in the following stochastic differential equation form
\begin{align}
d{\mathcal{X}}_t&=\frac{1}{C}(D_t-B_t)dt + \mathcal{X}_t(1-\mathcal{X}_t)\sigma_xdW_t,\label{eq:1}\\
\delta_t&=\ell(f(\mathcal{X}_t)+g(u_t)),\label{eq:2}\\
D_t&=B_t+\delta_t\Delta(\vmathbb{1}_{\delta_t>0}(1-B_t)+\vmathbb{1}_{\delta_t<0}B_t),\label{eq:3}\\
Y_t&=D_t+\sigma_Y\epsilon_t,\label{eq:4}
\end{align}
where $\mathcal{X}\in [0,\ 1]\subset \mathbb{R}$ is the state of charge, $B\in [0,\ 1]\subset \mathbb{R}$ is the baseline demand, $u\in [0,\ 1]\subset \mathbb{R}$ is the energy price, $\delta \in [0,\ 1]\subset \mathbb{R}$ is the demand change, $D\in [0,\ 1]\subset \mathbb{R}$ is the expected demand, $C$ is the capacity of flexible energy, $\Delta\in [0,\ 1]$ is the proportion of flexible demand, and $Y\in [0,\ 1]\subset \mathbb{R}$ is the demand output of the flexibility function. The above equations are constructed based on the normalised parameters between 0 and 1. $W$ is a Wiener process, i.e., $dW \sim N(0, dt)$, $\epsilon \sim N(0,1)$, $\sigma_X$ represent process noise intensity and $\sigma_Y$ is the standard deviation of the measurement noise. Moreover, the function $\vmathbb{1}_{\delta_t<0}$ is equal to 1 when $\delta_t<0$ and 0 otherwise, and the function $\vmathbb{1}_{\delta_t>0}$ is equal to 1 when $\delta_t>0$ and 0 otherwise. The nonlinear functions involved in the flexibility function are given by
\begin{align}
g(u) &= \beta_1 Is_1(u) + ~\cdots~ + \beta_7 Is_7(u),\label{eq:3y}\\
f(\mathcal{X}) &= (1-2\mathcal{X}+\alpha_1(1-(2\mathcal{X}-1)^2))(\alpha_2+\alpha_3(2\mathcal{X}-1)^2\notag \\
&+\alpha_4(2\mathcal{X}-1)^6),\label{eq:4y}\\
\ell(f(\mathcal{X})+ &g(u)) = -1 + \frac{2}{1+e^{-k(f(\mathcal{X})+g(u))}},\label{eq:5y}
\end{align}
where $Is_1$, \ldots, $Is_7$ are I-spline functions \cite{Ramsay1988}, and the parameters $\beta_1, \ldots, \beta_7$, $\alpha_1, \ldots, \alpha_4$, and $k$ are assumed to be unknown. They can be identified using different approaches, e.g., by maximizing the likelihood of observing the the actual measurements \cite{Tohidi22b}. By design, the functions $f(.)$ and $g(.)$ are monotonically decreasing and $\ell(.)$ is monotonically increasing. 

In the sequel, we first disregard the diffusion term in \eqref{eq:1}--\eqref{eq:4} and linearize the deterministic flexibility function in Section \ref{sec:linFF}. Then, design two controllers are designed to generate price such that the demand follows its reference. In Section \ref{sec:OPG}, the controller is designed assuming that the parameters are known. In Section \ref{sec:AFF}, we design an adaptive price generator when some parameters of the system are not known.


\section{DETERMINISTIC LINEARIZED FLEXIBILITY FUNCTION}\label{sec:linFF}

Disregarding the diffusion term in the flexibility function provides a simplified overview of the demand-response behaviour. This simplified version of flexibility function helps us understand the core of the dynamics of price-demand mapping. To this end, the flexibility function \eqref{eq:1}--\eqref{eq:4} can be modified as
\begin{align}
d\mathcal{X}_t/dt&=\frac{1}{C}(D_t-B_t),\label{eq:1x}\\
\delta_t&=\ell(f(\mathcal{X}_t)+g(u_t)),\label{eq:2x}\\
D_t&=B_t+\delta_t\Delta(\vmathbb{1}_{\delta_t>0}(1-B_t)+\vmathbb{1}_{\delta_t<0}B_t),\label{eq:3x}
\end{align}
where the output is $D_t$.
Although, the values of $B_t$, $D_t$, $u_t$ and $\mathcal{X}_t$ are between 0 and 1, the extreme cases of being 0 or 1 are less common. Having this fact in mind and considering an example of the nonlinear functions of the flexibility function, $f(\mathcal{X})$, $g(u)$, and $\ell(f(\mathcal{X})+g(u))$, in Figure 1, one can assume that $f(\mathcal{X})$, $g(u)$, and $\ell(f(\mathcal{X})+g(u))$ behave linearly in the range $[\epsilon, 1-\epsilon]$ for some $\epsilon$, $0\leq \epsilon < 1$. One can find the slope of $f(\mathcal{X})$, $g(u)$, and $\ell(f(\mathcal{X})+g(u))$ around a point in $[\epsilon, 1-\epsilon]$, as $\eta_1$, $\eta_2$, and $\eta_3$, respectively. Also, the biases for $f(\mathcal{X})$ and $g(u)$ can be considered as $\lambda_1$ and $\lambda_2$, respectively. Using the linearized version of the mentioned functions, (\ref{eq:1x})--(\ref{eq:3x}) can be rewritten as
\begin{align}
\hspace{-0.3cm} d\mathcal{X}_t/dt&=\frac{\Delta}{C}\eta_3(\eta_1\mathcal{X}_t+\lambda_1+\eta_2u_t+\lambda_2)\notag \\
&\times (\vmathbb{1}_{\delta_t>0}(1-B_t)+\vmathbb{1}_{\delta_t<0}B_t),\label{eq:1xx}\\
\delta_t&=\eta_3(\eta_1\mathcal{X}_t+\lambda_1+\eta_2u_t+\lambda_2),\label{eq:2xx}\\
D_t&=B_t+\delta_t\Delta(\vmathbb{1}_{\delta_t>0}(1-B_t)+\vmathbb{1}_{\delta_t<0}B_t).\label{eq:3xx}
\end{align}
The state equation (\ref{eq:1xx}) can be rewritten as the following piecewise defined function
\begin{align}\label{eq:5}
&\hspace{-0.2cm}d\mathcal{X}_t/dt =\notag  \\
&\hspace{-0.3cm}\left\{
        \begin{array}{ll}
            \frac{\Delta}{C}\eta_3(\eta_1\mathcal{X}_t+\eta_2u_t+\lambda_3)(1-B_t), &  \eta_3(\eta_1\mathcal{X}_t+\eta_2u_t+\lambda_3)>0, \\
            \frac{\Delta}{C}\eta_3(\eta_1\mathcal{X}_t+\eta_2u_t+\lambda_3)B_t, &  \eta_3(\eta_1\mathcal{X}_t+\eta_2u_t+\lambda_3)<0,
        \end{array}
    \right.
\end{align}
where $\lambda_3 = \lambda_1 + \lambda_2$. This implies that the state and output dynamics can be written in the form of a linear time-varying system as
\begin{align}
d\mathcal{X}_t/dt&=a_t\mathcal{X}_t + b_t u_t + d_t,\label{eq:6} \\
D_t &= B_t + Ca_t\mathcal{X}_t + Cb_t u_t + Cd_t,\label{eq:6xc}
\end{align}
where $a_t$ and $b_t$ are defined as
\begin{equation}\label{eq:7}
a_t = \left\{
        \begin{array}{ll}
            \frac{\Delta}{C}\eta_3\eta_1(1-B_t), & \ \eta_3(\eta_1\mathcal{X}_t+\eta_2u_t+\lambda_3)>0, \\
            \frac{\Delta}{C}\eta_3\eta_1B_t, & \ \eta_3(\eta_1\mathcal{X}_t+\eta_2u_t+\lambda_3)<0,
        \end{array}
    \right.
\end{equation}
\begin{equation}\label{eq:8}
b_t = \left\{
        \begin{array}{ll}
            \frac{\Delta}{C}\eta_3\eta_2(1-B_t), & \ \eta_3(\eta_1\mathcal{X}_t+\eta_2u_t+\lambda_3)>0, \\
            \frac{\Delta}{C}\eta_3\eta_2B_t, & \ \eta_3(\eta_1\mathcal{X}_t+\eta_2u_t+\lambda_3)<0,
        \end{array}
    \right.
\end{equation}
and
\begin{equation}\label{eq:8x}
d_t = \left\{
        \begin{array}{ll}
            \frac{\Delta}{C}\eta_3\lambda_3 (1-B_t), & \ \eta_3(\eta_1\mathcal{X}_t+\eta_2u_t+\lambda_3)>0, \\
            \frac{\Delta}{C}\eta_3\lambda_3 B_t, & \ \eta_3(\eta_1\mathcal{X}_t+\eta_2u_t+\lambda_3)<0.
        \end{array}
    \right.
\end{equation}
If $\eta_3(\eta_1\mathcal{X}_t+\eta_2u_t+\lambda_3)=0$, then $a_t = b_t = d_t=0$.
\begin{rem}\label{rem:1}
    By design, the functions $f$ and $g$ are monotonously decreasing and $\ell$ is monotonously increasing. Therefore, $\eta_1<0$, $\eta_2<0$, and $\eta_3>0$ and $a_t$ and $b_t$ are negative for all $t\geq 0$. Furthermore, $\lambda_1$ and $\lambda_2$ are positive scalars. The negativity of $a_t$ is important to the stability analysis of the linearized flexibility function. Also, having information of the sign of $b_t$ is required for the adaptive flexibility function design and will be utilized in Section \ref{sec:AFF}.
\end{rem}

\section{OPTIMAL PRICE GENERATOR}\label{sec:OPG}
Suppose that the scalar $C$ is known and that the values of the scalar functions $a_t$ and $b_t$ are known for all $t\geq 0$. Then, by setting $D_t = D_{ref_t}$ and isolating the input in \eqref{eq:6xc}, it can be shown that applying the price signal
\begin{align}\label{eq:9}
    u_t = \frac{1}{Cb_t}\left( -Ca_t\mathcal{X}_t-Cd_t-B_t+D_{{ref}_t}\right)
\end{align}
to (\ref{eq:6}) ensures that the output of the flexibility function, $D_t$, is equal to the reference demand signal, $D_{{ref}_t}$. It is noted that $Cb_t$ has a nonzero value for all $t\geq 0$. However, (\ref{eq:9}) may not be feasible since it does not consider the restrictions of $u_t$, that is, $u_t\in [0,\ 1]$. To this end, an optimization problem needs to be solved, so that minimizes $u_t - \frac{1}{Cb_t}( -Ca_t\mathcal{X}_t-Cd_t-B_t+D_{{ref}_t})$.

\begin{rem}\label{rem:2}
    Assume that the values of $a_t$, $b_t$, and $C$ are known for all $t\geq 0$, and that the signals $D_{{ref}_t}$ and $B_t$ are provided for all $t\geq 0$. Then, the optimization problem 
    \begin{align}\label{eq:10}
        min_{u_t} \ \ &\int_{t_k}^{t_{k+1}}(Cb_tu_t + Ca_t \mathcal{X}_t + Cd_t + B_t - D_{{ref}_t})^2,\notag \\
        d\mathcal{X}_t/dt&=a_t\mathcal{X}_t + b_t u_t + d_t,\notag \\
        u_t &\in [0,\ 1],
    \end{align}
    finds the optimal price signal for $t\in [t_k, \ t_{k+1})$. It is noted that $u_t$ is considered constant throughout the interval $[t_k, \ t_{k+1})$.
\end{rem}

\begin{rem}\label{rem:2x}
    Once a daily demand is purchased by an aggregator, and the daily baseline demand is provided, the bounds of integral in the optimization problem (\ref{eq:10}) can be extended and an optimal price signal can be calculated for the whole day.
\end{rem}

The procedure for implementing the proposed method of Section \ref{sec:linFF} is given in Implementation procedure 1.

\begin{algorithm}\label{alg:1}
\caption{Price signal generation algorithm for linearized flexibility function}
\noindent
\textbf{-Given $D_{{ref}_t}$, $B_t$, $C$, $\eta_1$, $\eta_2$, $\eta_3$, $\lambda_1$ and $\lambda_2$}

\textbf{-Calculate $\mathcal{X}_t$ using (\ref{eq:1x})} 


\textbf{-Solve the optimization algorithm (\ref{eq:10})}
\end{algorithm}


\begin{figure}
\label{fig:stable}
\centering
\includegraphics[width=0.5\textwidth]{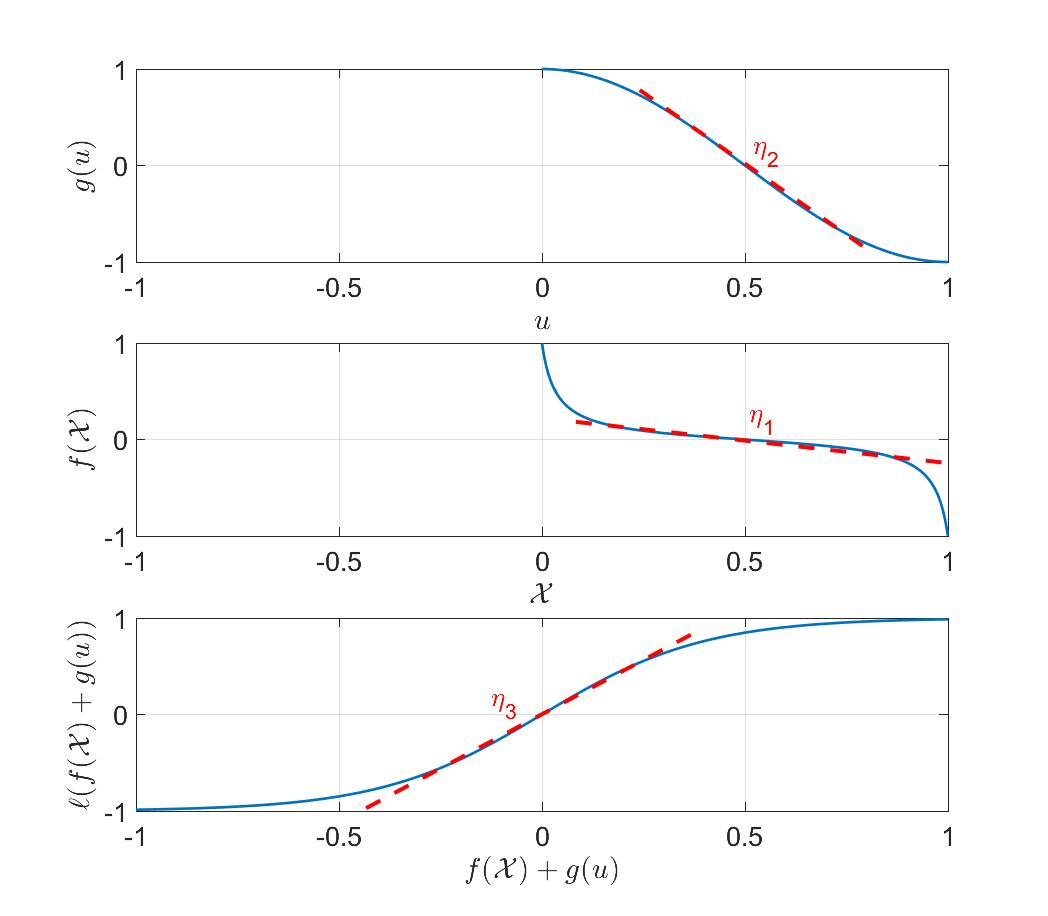}
\caption{Schematic of nonlinear functions of the flexibility function and their linear approximations.}
\end{figure}

\section{ADAPTIVE FLEXIBILITY FUNCTION} \label{sec:AFF}

In this section, we describe an approach for computing a control signal, $u_t$, when $a_t$ and $b_t$ are unknown, such that the demand, $D_t$, converges to its reference value, $D_{{ref}_t}$. Assume that $a_t$ and $b_t$ are piecewise constant.

Rewrite the state dynamics of the linearized flexibility function (\ref{eq:6}) as
\begin{align}
d\mathcal{X}_t/dt&=a\mathcal{X}_t + b (u_t + \bar{d}),\label{eq:6x}
\end{align}
where $\bar{d}=\lambda_3/ \eta_2$.

Let the reference dynamics be defined as 
\begin{align}
d\mathcal{Y}_t/dt&= \lambda \mathcal{Y}_t + \frac{1}{C}r_t,\label{eq:11}
\end{align}
where $\mathcal{Y}_t$ is the state of the reference dynamics and $r_t=D_{{ref}_t}-B_t$. Notice that the dynamics (\ref{eq:11}) is selected such that $\mathcal{Y}_t$ mimics the behaviour of (\ref{eq:1x}) when $D_t = D_{{ref}_t}$. This can be done by choosing a negative $\lambda$. The negativity of $\lambda$ ensures that the reference dynamics is stable.

A controller has to be designed such that it captures the changes of $r_t$ and $\mathcal{X}$, and generates a stabilizing control signal. Thus, we employ the control law
\begin{align}
u_t=\hat{\alpha}_t\mathcal{X}_t + \hat{\beta}_t r_t + \hat{\zeta}_t ,\label{eq:12}
\end{align}
where $\hat{\alpha}_t$, $\hat{\beta}_t$ and $\hat{\zeta}_t$ are control gains. With the control law (\ref{eq:12}), the closed loop dynamics can be written as
\begin{align}
d\mathcal{X}_t/dt&=(a+b\hat{\alpha}_t)\mathcal{X}_t + b\hat{\beta}_t r_t+ b(\hat{\zeta}_t+\bar{d} ).\label{eq:13}
\end{align}
If the flexibility function parameters were known, the ideal parameters could be calculated by comparing the closed-loop dynamics and the reference dynamics, i.e., $\alpha^* =\frac{\lambda-a}{b} $, $\beta^* = \frac{1}{bC}$ and $\zeta^* = -\bar{d}$.

By defining the error as $e_t = \mathcal{X}_t - \mathcal{Y}_t$, the error dynamics can be obtained as
\begin{align}
de_t/dt&=(a+b\hat{\alpha}_t)\mathcal{X}_t + b\hat{\beta}_t r_t + b(\hat{\zeta}_t+\bar{d}) -\lambda \mathcal{Y}_t - \frac{1}{C}r_t.\label{eq:14x}
\end{align}
Defining the parameter errors as $\tilde{\alpha}_t = \hat{\alpha}_t-\alpha^*$, $\tilde{\beta}_t = \hat{\beta}_t-\beta^*$ and $\tilde{\zeta}_t = \hat{\zeta}_t-\zeta^*$, the error dynamics (\ref{eq:14x}) can be rewritten as
\begin{align}
de_t/dt&=\lambda e_t + b\tilde{\alpha}_t\mathcal{X}_t +b\tilde{\beta}_tr_t + b\tilde{\zeta}_t.\label{eq:14}
\end{align}

In order to keep the parameters of the adaptive system bounded, a projection algorithm can be used. Here, we first define the projection algorithm, and then, introduce two useful lemmas in this regard.

\begin{defin}
    The projection operator, denoted as $ \text{Proj} $, for two scalars $ \theta $ and $Y$ is defined as 
\begin{align} \label{eq:17n}
&\text{Proj}(\theta, Y)\equiv \left \{\begin{array}{l}Y-Yh(\theta)\ \ if\ h(\theta)>0 \  \&  \ Y\left( \frac{dh(\theta)}{d\theta}\right) >0\\  Y \ \ \ \ \ \ \ \ \ \ \ \ \ \ \ \ \  \text{otherwise},\end{array}\right.
\end{align}
where $h(.):\mathbb{R}\rightarrow \mathbb{R}$ is a convex function defined as
\begin{equation} \label{eq:18n}
h(\theta)=\frac{(\theta-\theta_{min}-\varepsilon_{\theta})(\theta-\theta_{max}+\varepsilon_{\theta})}{(\theta_{max}-\theta_{min}-\varepsilon_{\theta})\varepsilon_{\theta}},
\end{equation}
where $ \varepsilon_{\theta} $ is the projection tolerance that should be chosen as $0<\varepsilon_{\theta}<0.5(\theta_{max}-\theta_{min})$. Also, $\theta_{max}$ and $\theta_{min}$ are the upper and lower bound of $\theta$. These bounds also form the projection boundary. In the convex function (\ref{eq:18n}), $h(\theta)=0$ when $\theta=\theta_{max}-\varepsilon_{\theta}$ or $\theta=\theta_{min}+\varepsilon_{\theta}$, and $h(\theta)=1$ when $\theta=\theta_{max}$ or $\theta_{min}$.

\begin{lemma}\label{lem:1}
 If $\dot{\theta}=\text{Proj}(\theta,Y)$ with initial conditions $\theta(0)\in \Omega_{\theta}=\{\theta\in \mathbb{R}|h(\theta)\leq 1\}$, where $h(\theta):\mathbb{R}\rightarrow \mathbb{R}$ is a convex function, then $\theta\in \Omega_{\theta}$ for $\forall t\geq 0$.   
\end{lemma}

\begin{proof}
    The proof of Lemma 1 can be found in \cite{Lavretsky13}.
\end{proof} 

\begin{lemma}\label{lem:2}
   For the ideal parameter $ \theta^*\in [\theta_{min}+\varepsilon_{\theta},\ \theta_{max}-\varepsilon_{\theta}] $, $ \theta\in \mathbb{R} $, $ Y\in \mathbb{R} $ and the projection algorithm in (\ref{eq:17n}) and (\ref{eq:18n}), the following inequality holds:
\begin{equation}\label{eq:e20}
\begin{array}{ll}
 ({\theta}^T-{{\theta}^*}^T) \big( -Y + \text{Proj}(\theta, Y) \big)  \leq 0.
\end{array}
\end{equation}

\end{lemma}

\begin{proof}
    The proof of Lemma 2 can be found in \cite{Tohidi19}.
\end{proof}
\end{defin}

The following theorem provides the main results of this paper. It provides the projection based adaptive laws along with stability analysis and convergence results.
\begin{thm} \label{thm:1}
    Consider the flexibility function dynamics (\ref{eq:6}) and the reference model (\ref{eq:11}), and assume that $a_t$ and $b_t$ are piecewise constant unknown parameters, but the sign of $b_t$ is considered to be known. Suppose that the price signal $u_t$, given in (\ref{eq:12}), is the control input of the flexibility function dynamics (\ref{eq:6})--(\ref{eq:6xc}) with the adaptive parameters, $\alpha_t$, $\beta_t$ and $\zeta_t$, that are updated using the following projection-based adaptive laws,
\begin{align}
d\hat{\alpha}_t/dt &= \gamma_{\alpha} \text{Proj}\left(\hat{\alpha}, -sgn(b_t)\mathcal{X}_te_t\right),\label{eq:15}\\
d\hat{\beta}_t/dt &= \gamma_{\beta} \text{Proj}\left(\hat{\beta}, -sgn(b_t)r_te_t\right),\label{eq:16}\\
d\hat{\zeta}_t/dt &= \gamma_{\zeta} \text{Proj}\left(\hat{\zeta}, -sgn(b_t)e_t\right),\label{eq:16xy}
\end{align}
where the projection operator ‘‘Proj’’ is defined in (\ref{eq:17n}), with convex function $h \in C^1$ in (\ref{eq:18n}), and $\gamma_{\alpha}$, $\gamma_{\beta}$ and $\gamma_{\zeta}$ are three positive adaptation gains. Then given any initial condition $e_0 \in \mathbb{R}$, $\alpha_0 \in \Omega_{\alpha}$, $\beta_0 \in \Omega_{\beta}$, $\zeta_0 \in \Omega_{\zeta}$, and $\alpha^* \in [\alpha_{min}+\varepsilon_{\alpha},\ \alpha_{max}-\varepsilon_{\alpha}]$, $\beta^* \in [\beta_{min}+\varepsilon_{\beta},\ \beta_{max}-\varepsilon_{\beta}]$ and $\zeta^* \in [\zeta_{min}+\varepsilon_{\beta},\ \zeta_{max}-\varepsilon_{\beta}]$, $\tilde{\alpha}_t$, $\tilde{\beta}_t$ and $\tilde{\zeta}_t$ remain uniformly bounded for all $t \geq 0$ and $e_t$ converges to $0$ as $t\rightarrow \infty$. Furthermore, $u_t$ remains bounded and $D_t$ converges to $D_{{ref}_t}$.
\end{thm}
\begin{proof}
    Consider the candidate Lyapunov function
\begin{align}
V(e_t,&\tilde{\alpha}_t,\tilde{\beta}_t,\tilde{\zeta}_t)\notag \\
&=\frac{1}{2}e_t^2+\frac{1}{2\gamma_{\alpha}}|b_t|\tilde{\alpha}_t^2+\frac{1}{2\gamma_{\beta}}|b_t|\tilde{\beta}_t^2 +\frac{1}{2\gamma_{\zeta}}|b_t|\tilde{\zeta}_t^2.\label{eq:17}
\end{align}
The time derivative of (\ref{eq:17}) along the trajectories of (\ref{eq:14}) and (\ref{eq:15})–(\ref{eq:16xy}) can be calculated as
\begin{align}
&dV(e_t,\tilde{\alpha}_t,\tilde{\beta}_t,\tilde{\zeta}_t)/dt\notag \\
&=\lambda e_t^2 + \left(b_t\mathcal{X}_te_t+|b_t|\text{Proj}\left(\hat{\alpha}, -sgn(b_t)\mathcal{X}_te_t\right)\right)\tilde{\alpha}_t\notag \\
&+ \left(b_tr_te_t+|b_t|\text{Proj}\left(\hat{\beta}, -sgn(b_t)r_te_t\right)\right)\tilde{\beta}_t\notag \\
&+ \left(b_te_t+|b_t|\text{Proj}\left(\hat{\zeta}, -sgn(b_t)e_t\right)\right)\tilde{\zeta}_t\notag \\
&=\lambda e^2 + |b_t|\left(sgn(b_t)\mathcal{X}_te_t+\text{Proj}\left(\hat{\alpha}, -sgn(b_t)\mathcal{X}_te_t\right)\right)\tilde{\alpha}_t\notag \\
&+ |b_t|\left(sgn(b_t)r_te_t+\text{Proj}\left(\hat{\beta}, -sgn(b_t)r_te_t\right)\right)\tilde{\beta}_t\notag \\
&+ |b_t|\left(sgn(b_t)e_t+\text{Proj}\left(\hat{\zeta}, -sgn(b_t)e_t\right)\right)\tilde{\zeta}_t.\label{eq:18}
\end{align}
The inequality (\ref{eq:e20}), introduced in Lemma \ref{lem:2}, implies that
\begin{align}
&dV(e_t,\tilde{\alpha}_t,\tilde{\beta}_t, \tilde{\zeta}_t)/dt\leq \lambda e^2\leq 0.\label{eq:19}
\end{align}
The negativity of $dV/dt$ implies that $e_t$, $\tilde{\alpha}_t$, $\tilde{\beta}_t$ and $\tilde{\zeta}_t$ are bounded, which causes $de_t/dt$ to be bounded as well. It also implies that 
\begin{align}
\int_0^t|\lambda|e_t^2dt \leq -\int_0^t (dV/dt)\ dt=V_0-V_t\leq V_0, \label{eq:20}
\end{align}
for all $t\geq 0$, which shows that $e_t\in \mathcal{L}_2$. Given that $e_t \in \mathcal{L}_2 \cap \mathcal{L}_{\infty}$, and $de_t/dt \in \mathcal{L}_{\infty}$, and using Barbalat's lemma, one can confirm that $\lim_{t\rightarrow \infty}e_t = 0$. Therefore, $\mathcal{X}_t$ converges to $\mathcal{Y}_t$, and since $\mathcal{Y}_t$, by design, follows the trajectories of (\ref{eq:1x}) with $D_t = D_{{ref}_t}$, the same holds for $\mathcal{X}_t$. Therefore, the adaptive price signal $u_t$, with the adaptation laws (\ref{eq:15}) and (\ref{eq:16}), leads to the convergence of $D_t$ to $ D_{{ref}_t}$. Moreover, by knowing the range of change of $\mathcal{X}_t$ and $r_t$, and by the selection of the upper and lower bounds of the projection algorithm one can guarantee that $u_t$ is bounded and is kept in the range $[0,\ 1]$. 
\end{proof}

\begin{rem}\label{rem:3}
    As described in Remark \ref{rem:1}, the sign of $b_t$ is always negative by design. Thus, the adaptation laws (\ref{eq:15})--(\ref{eq:16xy}) can be rewritten as
    \begin{align}
d\hat{\alpha}_t/dt &= \gamma_{\alpha} \text{Proj}\left(\hat{\alpha}, \mathcal{X}_te_t\right),\label{eq:15x}\\
d\hat{\beta}_t/dt &= \gamma_{\beta} \text{Proj}\left(\hat{\beta}, r_te_t\right),\label{eq:16x}\\
d\hat{\zeta}_t/dt &= \gamma_{\zeta} \text{Proj}\left(\hat{\zeta}, e_t\right).\label{eq:16xyz}
\end{align}
\end{rem}

\begin{rem}
    The projection bounds should be selected such that $u_t$ is bounded in the range $[0,\ 1]$. Suppose that the reference demand is always greater than or equal to the baseline demand, by choosing $\alpha_{min}=0$, $\alpha_{max}=1/3$, $\beta_{min}=0$, $\beta_{max}=1/3$, $\zeta_{min}=0$ and $\zeta_{max}=1/3$, one can limit $u_t$ in the range $[0,\ 1]$.
\end{rem}

\begin{rem}\label{rem:4}
    In order to implement the adaptation laws (\ref{eq:15x})--(\ref{eq:16xyz}), one requires to know the values of $\mathcal{X}_t$. Having the demand or its estimation, $D_t$, the baseline, $B_t$ and the flexibility capacity, $C$, $\mathcal{X}_t$ can be calculated using (\ref{eq:1x}) as $\mathcal{X}_t = \int_{0}^{t}\frac{1}{C}(D_t-B_t)dt$.
\end{rem}

\begin{rem}\label{rem:5}
    Every day, an aggregator purchases a specific amount of energy for each hour, $\bar{D}_{ref}$. If an estimate of the hourly demand, $\bar{D}$, of a price-responsive energy system is available along with the hourly baseline, $\bar{B}$, one can implement the proposed adaptive approach and find the hourly price signal throughout each day. This hourly price signal can then be used in a model predictive controller or an energy management system (EMS).
\end{rem}

A block diagram of the proposed adaptive flexibility function is shown in Figure 2. In addition, the procedure for implementing the proposed method of Theorem \ref{thm:1} is given in Implementation procedure 2.

\begin{figure}
\label{fig:2}
\vspace{0.5cm}
\centering
\includegraphics[width=0.49\textwidth]{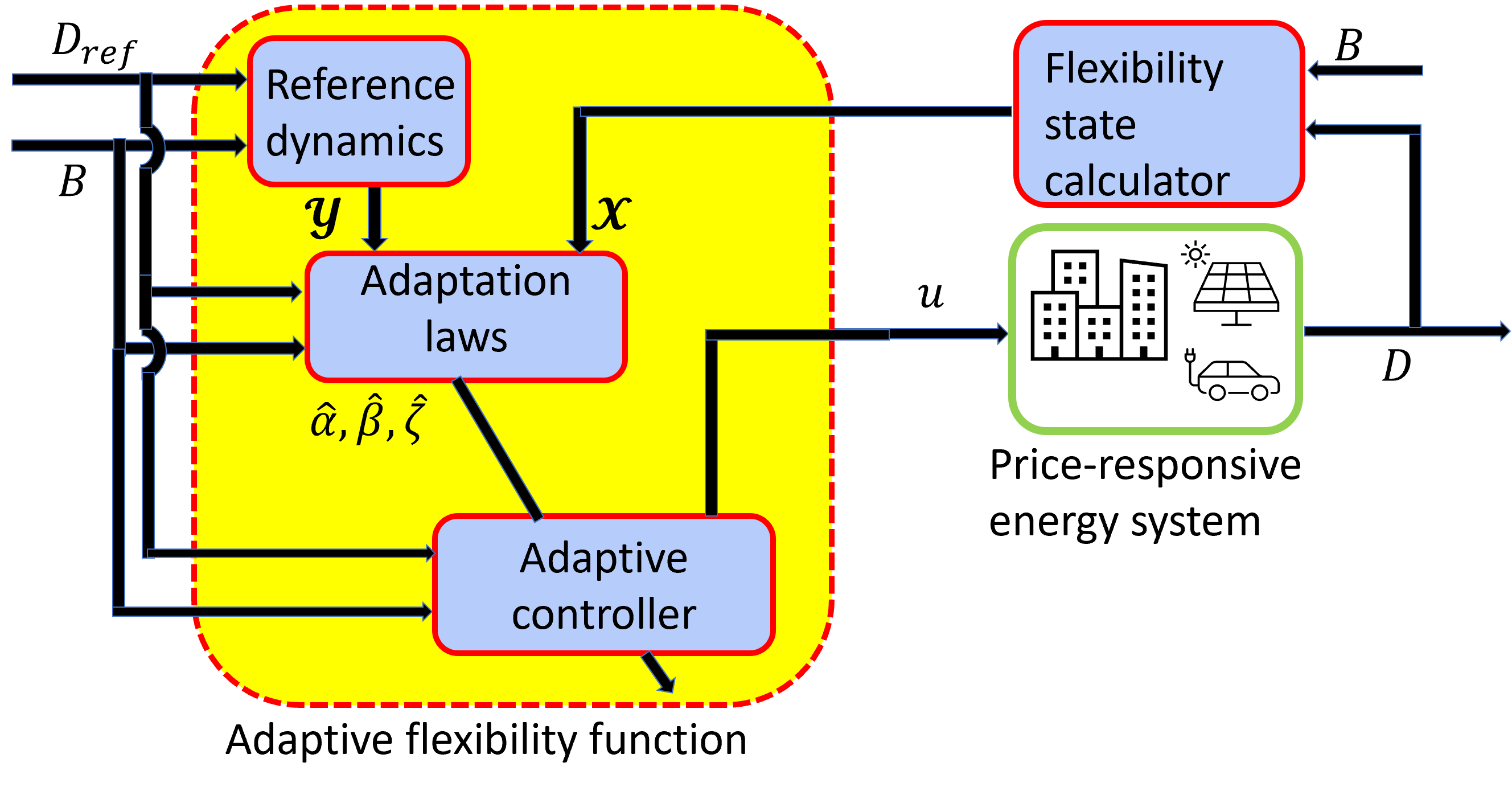}
\caption{Block diagram of the proposed adaptive flexibility function.}
\end{figure}

\begin{algorithm}\label{alg:2}
\caption{Adaptive price signal generation algorithm based on the linearized flexibility function}
\noindent
\textbf{-Given $D_t$, $D_{{ref}_t}$, $B_t$, $C$}

\textbf{-Set $\lambda$, $\gamma_{\alpha}$, $\gamma_{\beta}$ and $\gamma_{\zeta}$}

\textbf{-Set the projection algorithm parameters $\zeta_{\alpha}$, $\zeta_{\beta}$, $\alpha_{min}$, $\alpha_{max}$, $\beta_{min}$, $\beta_{max}$, $\zeta_{min}$ and $\zeta_{max}$}

\textbf{-Provide the reference dynamics $\mathcal{Y}_t$ using (\ref{eq:11})}

\textbf{-Calculate $\mathcal{X}_t$ using Remark \ref{rem:4}}

\textbf{-Calculate the error $e_t = \mathcal{X}_t - \mathcal{Y}_t$}

\textbf{-Implement the adaptation laws (\ref{eq:15x})--(\ref{eq:16xyz}) and find $\hat{\alpha}_t$, $\hat{\beta}_t$ and $\hat{\zeta}_t$}

\textbf{-Employ the control (\ref{eq:12})}

\end{algorithm}


\begin{figure}
\label{fig:2x}
\centering
\includegraphics[width=0.5\textwidth]{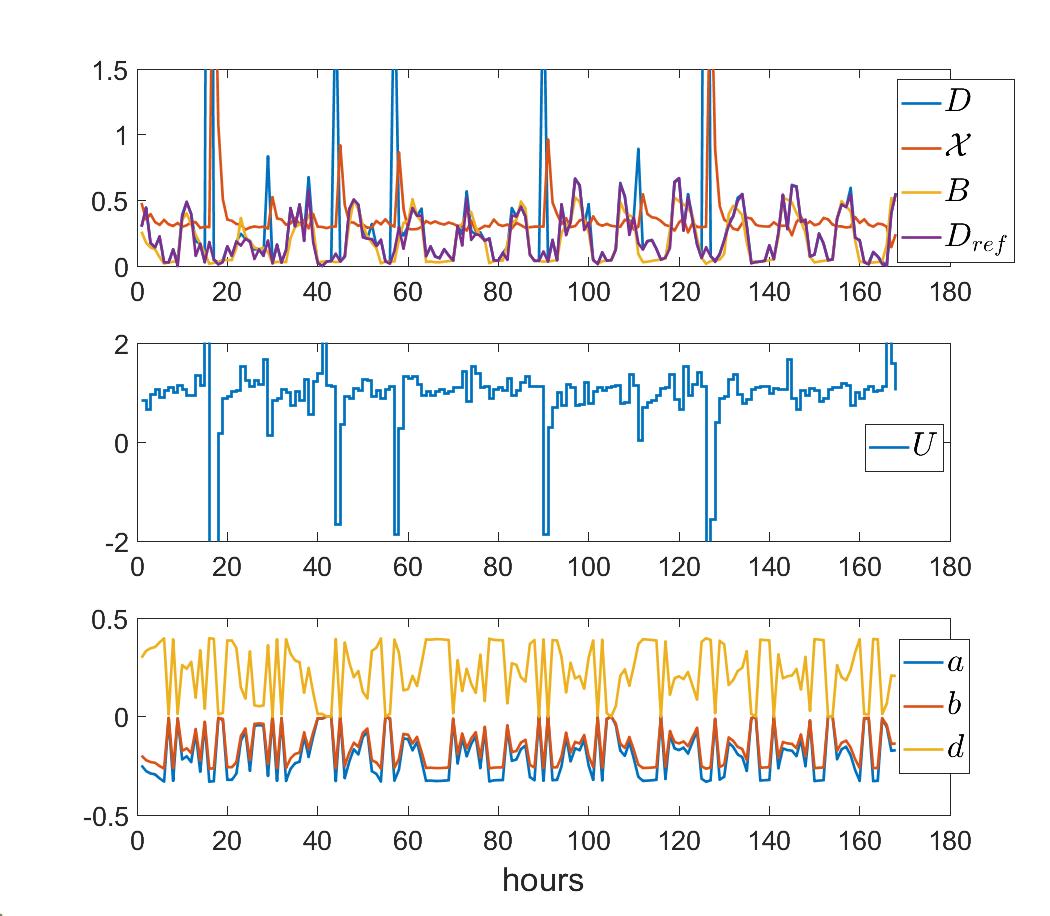}
\caption{Price signal generation using the controller (\ref{eq:9}) without bounds on the price signal.}
\end{figure}

\begin{figure}
\label{fig:3}
\centering
\includegraphics[width=0.5\textwidth]{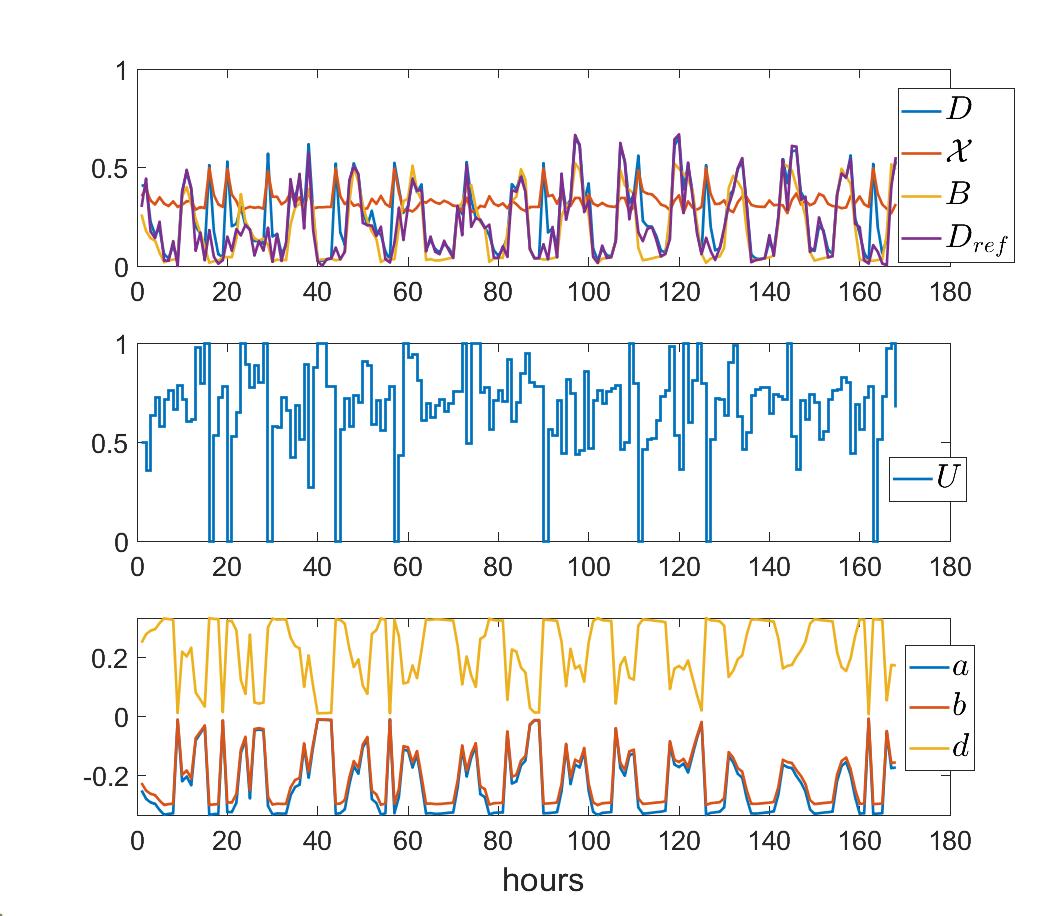}
\caption{Price signal generation using the controller (\ref{eq:9}) with bounds on the price signal.}
\end{figure}

\begin{figure}
\label{fig:4}
\centering
\includegraphics[width=0.5\textwidth]{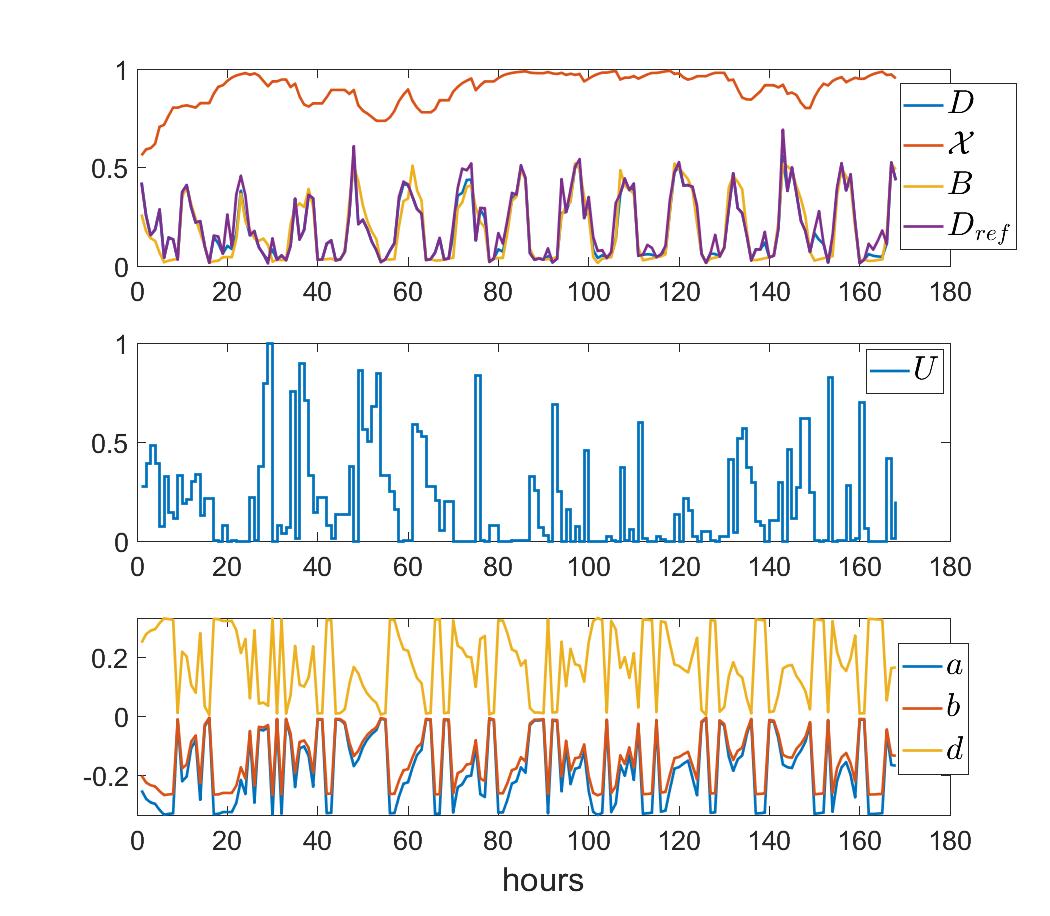}
\caption{Price signal generation using Implementation procedure 1.}
\end{figure}

\begin{figure}
\label{fig:5}
\centering
\includegraphics[width=0.5\textwidth]{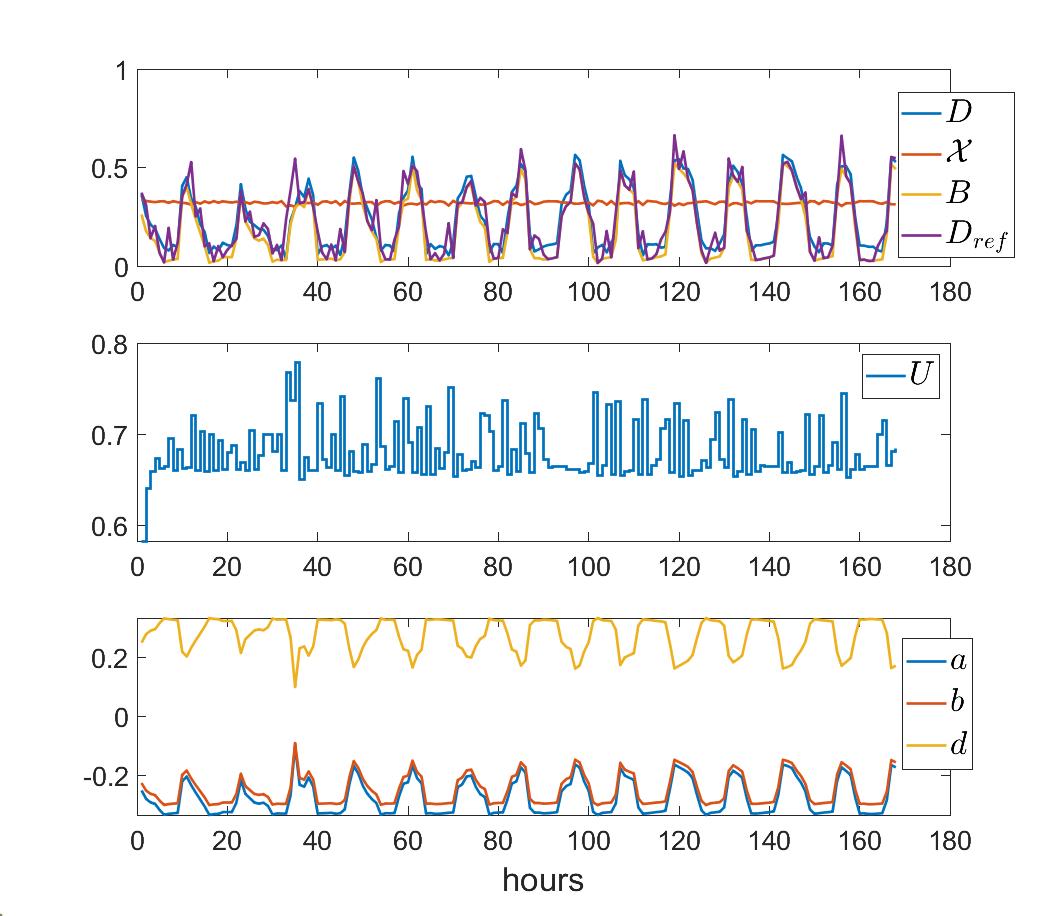}
\caption{Price signal generation using Implementation procedure 2.}
\end{figure}

\section{SIMULATION RESULTS}\label{sec:sim}
The linearized flexibility function is utilized to demonstrate the effectiveness of the proposed approaches. The parameters of the linearized flexibility function are $\eta_1 = -1$, $\eta_2 = -0.9$, $\eta_3 = 1$, $\lambda_1 = 0.5$, $\lambda_2 = 0.5$, $C=2.97$ and $\Delta = 1$.  

The first scenario considers all of the parameters to be known. Thus, we follow Implementation procedure 1. First, we implement (\ref{eq:9}) without solving optimization or any bound on the price signal. Figure 3 demonstrates the results of the linearized flexibility function. The top panel shows the flexibility state, $X$, its output, $D$, baseline signal, $B$, and the reference demand, $D_{ref}$. It is seen that $D$ follows $D_{ref}$ conveniently. However, there are some mismatch between these two at the switching times, i.e. when the sign of $\delta_t$ changes. Even at the switching times, it is seen that the control mechanism recovers and follows the reference demand shortly. The middle panel illustrates the generated price signal. It is seen that the generated price signal is not limited between 0 and 1. The third panel shows the time-varying parameters of the linearized flexibility function.

Figure 4 follows the first scenario without solving an optimization problem but with a software limitation. It is seen in the middle panel of this figure that the control signal, price signal, is limited between 0 and 1, using a software limitation. It is seen that bounding the price signal does not cause instability. It even pushes the state and output of the linearized flexibility function to the prespecified limit of 0 and 1, as can be seen in the top panel. Also, the demand follows its reference. The third panel shows the time-varying parameters of the linearized flexibility function.

The results of Implementation procedure 1, with optimization problem (\ref{eq:10}) is illustrated in Figure 5. Top panel of this figure, shows the results of the flexibility state, its output, the baseline signal, and the reference demand. It is seen that the demand  follows the reference demand. The middle panel shows the price signal. The optimization problem finds the optimal solution while considering the constraints. The time-varying parameters of the linearized system is demonstrated in the bottom panel.

The results of Implementation procedure 2, where the parameters of the flexibility function are not known, is illustrated in Figure 6. Top panel of this figure, shows the results of the flexibility state, its output, the baseline signal, and the reference demand. It is seen that demand  follows the reference demand. The middle panel shows the price signal. The optimization problem finds the optimal solution while considering the constraints. The time-varying parameters of the linearized system is demonstrated in the bottom panel.

\section{CONCLUSIONS}\label{sec:conc}
An adaptive flexibility function based on adaptive model reference controller structure is proposed in this paper. The method utilizes the linearized price-demand mapping and generates an adaptive price signal to diminish the difference between the demand and the reference demand in an energy system. The proposed method considers price signal constraints using projection algorithm. Furthermore, the method needs neither uncertainty identification nor persistence of excitation assumption. This property along with the adaptation capability simplifies offering the flexibility services e.g. in a plug-and-play manner and without the need to conduct a manual, customized modeling-and-control study for each resource separately. 
Simulation results show the effectiveness of the proposed method.





\section*{ACKNOWLEDGMENT}
This work is supported by Sustainable plus energy neighbourhoods (syn.ikia) (H2020 No. 869918), ELEXIA (Horizon Europe No. 101075656),  ARV (H2020 101036723), SEM4Cities (IFD Project No. 0143-0004), IEA EBC - Annex 81 - Data-Driven Smart Buildings (EUDP Project No. 64019-0539), and IEA EBC - Annex 82 - Energy Flexible Buildings Towards Resilient Low Carbon Energy Systems (EUDP Project No. 64020-2131).



\bibliographystyle{ieeetr}
\bibliography{main}

\end{document}